\documentclass[a4paper,UKenglish]{lipics-v2016}
 
\usepackage[numbers,sort&compress]{natbib}
\usepackage{microtype}
\usepackage{amstext,amssymb,amsmath}
\usepackage{xspace}
\usepackage{tikz}

\newcommand{\executeiffilenewer}[3]{%
\ifnum\pdfstrcmp{\pdffilemoddate{#1}}%
{\pdffilemoddate{#2}}>0%
{\immediate\write18{#3}}\fi%
} 

\newcommand{\svg}[2]{\def\svgwidth{#1}%
\executeiffilenewer{#2.svg}{#2.pdf}%
{inkscape -z -D --file=#2.svg %
--export-pdf=#2.pdf --export-latex}%
{\input{#2.pdf_tex}}}
\renewcommand{\st}{\mathrel : }
\newcommand{\LLL}{\mathcal{L}}
\newcommand{\SSS}{\ensuremath{\mathcal{S}}} 

\newcommand{\kcrouting}[2]{\textsc{$(#1,#2)$-Congestion Routing}\xspace}
\newenvironment{TheoremNo}[1]{\noindent{\bfseries Theorem #1.\ }\itshape}{\mbox{}\newline\rm}

\usepackage{xargs}
\usepackage[colorinlistoftodos,prependcaption,textsize=tiny]{todonotes}

\title{Routing with Congestion in Acyclic Digraphs}
\authorrunning{Saeed Akhoondian Amiri, Stephan Kreutzer, D\'aniel Marx,
  Roman Rabinovich}
\author[1]{Saeed Akhoondian Amiri}
\author[1]{Stephan Kreutzer\thanks{The research of Saeed Amiri,
    Stephan Kreutzer and Roman Rabinovich has been supported by the
European Research Council (ERC) under the European Union's Horizon
2020 research and innovation programme (ERC consolidator grant DISTRUCT,
agreement No 648527). The research of D\'aniel Marx was supported by
ERC Starting Grant PARAMTIGHT (No.~280152) and OTKA grant NK105645.}} 
\author[2]{D\'aniel Marx}
\author[1]{Roman Rabinovich}
\affil[1]{Technical University Berlin, 
  Sekr TEL 7-3, Ernst-Reuter-Platz 7, 
  10587 Berlin, Germany,\\
  \texttt{\{saeed.amiri,stephan.kreutzer,roman.rabinovich\}@tu-berlin.de}}
\affil[2]{Institute for Computer Science and Control,
Hungarian Academy of Sciences (MTA SZTAKI)
\texttt{dmarx@cs.bme.hu}.}

\begin{document}

\maketitle

\begin{abstract}
  We study the version of the $k$-disjoint paths problem where $k$
  demand pairs $(s_1,t_1)$, $\dots$, $(s_k,t_k)$ are specified in the
  input and the paths in the solution are allowed to intersect, but
  such that no vertex is on more than $c$ paths.  We show that on
  directed acyclic graphs the problem is solvable in time $n^{O(d)}$
  if we allow congestion $k-d$ for $k$ paths. Furthermore, we show
  that, under a suitable complexity theoretic assumption, the problem
  cannot be solved in time $f(k)n^{o(d/\log d)}$ for any computable
  function $f$.
\end{abstract}

 \section{Introduction}\label{sec:intro}

 The $k$-disjoint paths problem and related routing problems are among
 the central problems in combinatorial optimisation. In the most basic
 variant of the $k$-disjoint paths problem, a graph $G$ is given with
 $k$ pairs $(s_1, t_1)$, $\dots$, $(s_k, t_k)$ of vertices and the
 task is to find $k$ pairwise vertex-disjoint paths linking each
 $s_i$ to its corresponding target $t_i$.

The problem is well known to be NP-complete \cite{MR57:11691}. On
undirected graphs with a fixed number $k$ of source/terminal pairs,
Robertson and Seymour proved in their monumental graph minor series
\cite{GM-series} that the problem is polynomial-time
solvable. In fact, they showed that it is fixed-parameter
tractable with parameter $k$: it can be solved in cubic time for every
fixed value of $k$.

For directed graphs, the problem is computationally much harder. Fortune et
al.~\cite{FortuneHW80} proved that it is already NP-complete
for only $k=2$ source/terminal pairs. In particular, this also implies
that it is not fixed-parameter tractable on directed graphs. 
Following this result a lot of work has gone into establishing more
efficient algorithms on restricted classes of digraphs. 

Fortune et al.~\cite{FortuneHW80} showed that the problem can be
solved in time $n^{O(k)}$ on acyclic digraphs, that is, it is
polynomial-time for every fixed $k$. However, as proved by Slivkins
\cite{Slivkins10}, the problem is $W[1]$-hard on acyclic digraphs, and
therefore unlikely to be fixed-parameter tractable. On the other hand,
Cygan et al. \cite{CyganMPP13} proved that the problem is
fixed-parameter tractable with parmeter $k$ when restricted to planar
digraphs. Related to this, Amiri et al.~\cite{AmiriGolKreSie14} proved
that the problem remains NP-complete even in upward planar graphs, but
admits a single exponential fixed-parameter algorithm.

Disjoint paths problems have also been studied intensively in the area
of approximation algorithms, both on directed and undirected graphs (see,
e.g.,~\cite{ChekuriKS06,KolliopoulosS04,AndrewsCGKTZ10,ChekuriKS04,ChekuriKS05,Chuzhoy12,ChekuriE13,ChuzhoyLi12,ChekuriE14}). 
The goal is, given an input graph $G$ and demands
$(s_1, t_1), \dots, (s_k, t_k)$ to \emph{route} as many pairs as
possible in polynomial time. There are many variations what it means
for a pair to be routable. In particular, a problem studied
intensively in the
approximation literature is a relaxed version of disjoint paths where
the paths are no longer required to be fully disjoint. Instead, they
may intersect but every vertex of the graph is allowed to be contained
in at most $c$ paths, for some fixed constant $c$. This is called
\emph{congestion $c$ routing}. In particular, the well-linked decomposition
framework developed in \cite{ChuzhoyLi12} for undirected
graphs and later generalised to digraphs in \cite{ChekuriE14} has
proved to be very valuable for obtaining good approximation
algorithms for disjoint paths problems on planar graphs and digraphs.

In this paper, we are interested in exact solutions for high congestion
routing on acyclic digraphs. More precisely, we study the following problem.

\begin{definition}
  \begin{enumerate}
  \item Let $G$ be a digraph and let
    $I := \{ (s_1, t_1), \dots, (s_k, t_k) \}$ be a set of pairs of
    vertices. Let $c\geq 1$. A \emph{$c$-routing} of $I$ is a set
    $\{ P_1, \dots, P_k\}$ of paths such that, for all
    $1\leq i \leq k$, path $P_i$ links $s_i$ to $t_i$ and no vertex
    $v\in V(G)$ appears in more than $c$ paths from
    $\{P_1, \dots, P_k\}$.
  \item Let $k, c\geq 1$.  In the \kcrouting{k}{c} problem, a digraph $G$ is given in the input together with a
    set $I := \{(s_1, t_1), \dots, (s_k, t_k)\}$ of $k$ pairs of
    vertices (the demands); the task is to decide whether there is a
    $c$-routing of $I$ in $G$.
  \end{enumerate}
\end{definition}

We consider \kcrouting{k}{c} on acyclic digraphs. First, it is not
very difficult to show that, for every fixed $c\ge 1$, we can
generalize the $n^{O(1)}$ time algorithm of
Fortune~et~al.~\cite{FortuneHW80} to \kcrouting{k}{c}.  By revisiting
the W[1]-hardness proof of Slivkins~\cite{Slivkins10} and making appropriate
modifications, we can establish that the problem remains W[1]-hard for
every fixed congestion $c\ge 1$. Moreover, by doing the proof in a
more modern way (reducing from general subgraph isomorphism instead of
maximum clique and invoking a lower bound of Marx
\cite{marx-toc-treewidth}), we can show that the $n^{O(k)}$ time
algorithm is essentially best possible with respect to the exponent of
$n$. This lower bound is under the Exponential-Time Hypothesis (ETH),
which can be informally stated as $n$-variable \textsc{3Sat} cannot be
solved in time $2^{o(n)}$ (see
\cite{MR1894519,DBLP:journals/eatcs/LokshtanovMS11,DBLP:books/sp/CyganFKLMPPS15}
for more background).

\begin{theorem}\label{thm:hardness}
For any fixed integer $c\ge 1$, \kcrouting{k}{c} is \textup{W[1]}-hard parameterized by $k$ and, assuming ETH, cannot be solved in time $f(k)n^{o(k/\log k)}$ for any computable function $f$.
\end{theorem}

While Theorem~\ref{thm:hardness} shows that the problem does get
easier for any larger fixed congestion $c\ge 1$, intuitively one
expects the problem to get simpler at some point: after all, the
problem is trivial if $c\ge k$. Therefore, we study the complexity of
the problem in settings close to this extreme case.  
The main algorithmic result of this paper is to show that for any fixed value
of $d\ge 1$, the $\kcrouting{k}{k-d}$ problem can be solved in time
$n^{O(d)}$. That is, the exponent of the polynomial bounding the 
running time of the algorithm only depends on $d$ but not
on the number $k$. 

\begin{theorem}\label{thm:main-algo}
  For every fixed $d\ge 1$, the $\kcrouting{k}{k-d}$ problem on
  acyclic digraphs can be solved in time $n^{O(d)}$.
\end{theorem}

A simple corollary of Theorem~\ref{thm:hardness} shows that
\kcrouting{k}{k-d} is unlikely to be fixed-parameter tractable and the
running time of Theorem~\ref{thm:main-algo} essentially cannot be improved
(assuming ETH). Observe that if we set $d := k-1$, then
\kcrouting{k}{k-d} is simply the standard $k$-disjoint path problem,
thus any algorithmic result for \kcrouting{k}{k-d} parameterized by
$d$ would imply the essentially same algorithmic result for the fully
disjoint version parameterized by $k$. 

\begin{corollary}\label{cor:hard2}
\kcrouting{k}{k-d} is \textup{W[1]}-hard parameterized by $d$ (if $k$
is part of the input) and, assuming ETH, cannot be solved in time
$f(k)n^{o(d/\log d)}$ for any computable function $f$. 
\end{corollary}

\medskip\noindent\textbf{Organisation. } The paper is organised as follows. In
Section~\ref{sec:polytime} we fix some notation and prove our main
algorithmic result. The corresponding lower bound in then proved in
Section~\ref{sec:lowerbounds}.%  We conclude and discuss future research
% in Section~\ref{sec:conclusion}.

\section{Preliminaries}

We review basic notation and concepts of graph theory needed in the
paper. We refer to \cite{Diestel10,BangJensenG10} for background. 

Let $G$ be a digraph. We write $V(G)$ and $E(G)$ for its set of
vertices and edges, respectively. We assume that there is no edge with
the same head and tail, i.e.~there are no loops in the digraphs we
consider in this paper. If $(u,v)\in E(G)$ is an edge, then $u$ is its
\emph{tail} and $v$ is its \emph{head}. $G$ is \emph{simple} if there
are no two distinct edges which have the same tail and the same
head. Otherwise we call $G$ a \emph{multi digraph}.

A \emph{path} $P$ in a digraph $G$ is determined by a sequence $(v_1,
\dots, v_{\ell})$ of vertices such that $v_i\not=v_j$ for all $1 \leq i < j \leq \ell$
and $(v_i, v_{i+1})\in E(G)$ for all $1\leq i < \ell$. We write $E(P)$
for the set $\{ (v_i, v_{i+1}) \st 1\leq i \leq \ell-1\}$ of edges
appearing in $P$ and $V(P)$ for the set $\{ v_1, \dots, v_\ell\}$ of
vertices. We say that $P$ \emph{links}  $v_1$ to $v_{\ell}$.

Two paths $P_1$
and $P_2$ are \emph{edge disjoint} if $E(P_1)\cap E(P_2)=\emptyset$.

\section{A polynomial-time algorithm on acyclic digraphs}
\label{sec:polytime}

In this section we prove the first main result of this paper,
Theorem~\ref{thm:main-algo}, which we repeat here for convenience.

\medskip

% \begin{theorem}\label{thm:main-algo}
%   On acyclic digraphs the $\kcrouting{k}{k{-}d}$ problem can be solved in
%   time $n^{O(d)}$. 
% \end{theorem}

\begin{TheoremNo}{\ref{thm:main-algo}}
  For every fixed $d\ge 1$, the $\kcrouting{k}{k-d}$ problem on
  acyclic digraphs can be solved in time $n^{O(d)}$.
\end{TheoremNo}

We first need some additional notation and prove some auxiliary lemmas.

\begin{definition}
  Let $G$ be a digraph and let $\LLL$ be a set of paths in $G$. For
  every $v\in V(G)$ we define the \emph{congestion} of $v$ with
  respect to $\LLL$ as the number of paths in $\LLL$ containing $v$.
\end{definition}

The following lemma provides a simple extension of the algorithm
from~\cite{FortuneHW80} for disjoint paths in acyclic digraphs. 

\begin{lemma}\label{lem:constant-number-of-paths}
  On acyclic digraphs $G$ the $\kcrouting{k}{c}$ probem can be
  solved in time $n^{O(k)}$, where $n:= |G|$.
\end{lemma}
\begin{proof}
  In~\cite{FortuneHW80}, Fortune et al. proved that the $k$-disjoint
  paths problem can be solved in time $n^{O(k)}$ on any $n$-vertex acyclic
  digraph $G$. 

  Let $G$, $(s_1, t_1), \dots, (s_k, t_k)$ and $c$ be given. 
  We construct a new digraph $H$ with $V(H) := V(G)\times \{ 1, \dots,
  c\}$ and $E(H) := \{ \big( (u,i), (v, j)\big) \st (u,v)\in
  E(G)\}$. 

  Then $H$ contains $k$ pairwise vertex disjoint paths $P_1,
  \dots, P_k$ such that $P_i$ links $(s_i, 1)$ to $(t_i, 1)$ if, and
  only if, there is a positive solution to the $(k, c)$-Congestion
  Routing Problem on~$G$. By the algorithm in~\cite{FortuneHW80} we
  can decide whether the paths $P_1, \dots, P_k$ exist in $H$ in time
  $|V(H)|^{O(k)}$ and hence in time $(c\cdot n)^{O(k)} =  n^{O(k)}$ as $c\leq n$.
\end{proof}

We will use this lemma in the form given in the next corollary. 

\begin{corollary}\label{cor:constant-number-of-paths}
  For $c, k\geq 0$ such that $k\in O(c)$, the $\kcrouting{k}{c}$
  problem can be solved on any acyclic $n$-vertex digraph $G$ in time $n^{O(c)}$.
\end{corollary}

The next lemma provides the main reduction argument for proving Theorem~\ref{thm:main-algo}.

\begin{lemma}\label{lem:reducing-number-of-paths}
  Let $G$ be an acyclic directed graph and let $d \geq 1$ and $k>3d$. Let
  $ I := \{ (s_1,t_1), \dots, (s_k,t_k)\}\subseteq V(G)\times V(G)$ be a set of
  source/terminal pairs. There exists a $(k{-}d)$-routing of $I$ if,
  and only if, for
  every pair $(s, t)\in I$ there is a path in $G$ from $s$ to $t$ and
  there is a subset $I'\subsetneq I$ of order 
  $|I'| = k-1$ such that there is a $(k-d-1)$-routing of $I'$.
\end{lemma}
\begin{proof}
  The if direction is easy to see. Let $\SSS' := \{P_1, \dots,
  P_{k-1}\}$ be a $(k-d-1)$-routing of a set $I'\subseteq I$ of order
  $k-1$. Let $s, t$ be such that $I = I' \cup \{ (s, t)\}$. By
  assumption there is a simple path $P$ from $s$ to $t$ in
  $G$. Then $\SSS := \SSS' \cup \{ P\}$ is a $(k-d)$-routing of $I$. 

  For the reverse direction let $I := \{ (s_1, t_1), \dots, (s_k,
  t_k)\}$ and let $\hat\SSS := \{ \hat P_1, \dots, \hat P_k\}$ be
  a $(k-d)$-routing of $I$ such that $\hat P_i$ links $s_i$ to $t_i$, for
  all $1\leq i \leq k$.  
  We define a multi digraph $G'$ on the same vertex set $V(G)$ as $G$
  as follows. For every pair $u, v\in V(G')$ such that $e = (u,v)\in
  E(G)$ and every $1\leq i \leq k$, if $e$ occurs on the path $\hat P_i \in
  \SSS$, then we add a new edge $e^i = (u, v)$ to $G'$. Hence, if any
  edge $e\in E(G)$ is used by $\ell$ different paths in $\hat\SSS$, then
  $G'$ contains $\ell$ parallel edges between the endpoints of $e$. In
  the rest of the proof we will work on the multi digraph $G'$. We can
  now take a set $\SSS := \{ P_1, \dots, P_k\}$ of pairwise edge
  disjoint paths, where $P_i$ is the
  path from $s_i$ to $t_i$ induced by the edge set $\{ e^i \st e\in
  E(\hat{P}_i)\}$. That is, by using the parallel edges, we can turn
  the routing $\hat{\SSS}$ into a 
  $(k{-}d)$-routing $\SSS$ of $I$ where the
  paths are mutually edge disjoint. 

  In the remainder of the proof we will construct a subset
  $I'\subsetneq I$ of order $k-1$ and a $(k-d-1)$-routing of
  $I'$ in $G'$ which is pairwise edge disjoint. This naturally induces
  a $(k-d-1)$-routing of $I'$ in $G$. Note that in $G'$, if
  $\LLL$ is any set of pairwise edge disjoint paths, then the
  congestion of any vertex with respect to $\LLL$ is at most the
  congestion of the vertex with respect to $\SSS$ (and thus $\hat
  \SSS$) in $G'$ (and $G$, respectively). 

  Let $\sqsubseteq$ be a topological ordering of $G'$ and let 
  $A := \{ a_1, \dots, a_{\ell}\}$ be the set of vertices of congestion
  $k-d$ with respect to $\SSS$ such that $a_i \sqsubseteq a_j$ whenever
  $i<j$. As $k>3d$, for all $1\leq i < {\ell}$ there is a path in $G$ from
  $a_i$ to $a_{i+1}$.   
  
  For  $1\leq i \leq k$, an \emph{atomic subpath} of $P_i$ (with respect to
  $\SSS$) is a subpath of $P_i$ that starts and ends in a vertex of
  $A\cup \{ s_i, t_i\}$ and is internally vertex disjoint from $A$. 
  Hence, every path $P_i\in \SSS$ consists of the concatenation $P^1_1
  \cdot \dots \cdot P_i^{{\ell}_i}$ of its atomic subpaths where we identify the
  last vertex of $P^j_i$ with the first vertex of $P^{j+1}_i$ for all $1\le j <
  \ell_i$. Note that any
  two atomic subpaths of paths $P_i, P_j$ in $\SSS$ are pairwise edge
  disjoint. 
  
  Let $I' \subset I$ be a subset of order $k-1$. A routing $\SSS' :=
  \{ P'_1, \dots, P'_{k-1}\}$ of
  $I'$ is \emph{conservative with respect to $\SSS$} if it consists of
  pairwise edge disjoint paths and every path in
  $\SSS'$ consists of a concatenation of atomic subpaths of paths in
  $\SSS$. In the sequel,
  whenever we speak of a conservative $I'$-routing we implicitly mean
  that it is conservative with respect to $\SSS$.

  If $\SSS'$ is a conservative $I'$-routing with respect to
  $\SSS$, then it consists of pairwise edge disjoint paths and hence
  for every $v\in V(G)$ the congestion of $v$ with respect to $\SSS'$
  is at most the congestion of $v$ with respect to $\SSS$.

  Let $1\leq i_1 < i_2\leq {\ell}$ and let $1\leq j \leq k$. Let $\SSS'$ be
  a conservative $I'$-routing. An \emph{$(i_1, i_2)$-jump of colour $j$} is a subpath $P'$ of
  $P_j$ from $a_{i_1}$ to $a_{i_2}$ such that for all $i$ with $i_1 < i < i_2$
  the vertex $a_i$ is not on $P_j$. Note that any jump is an atomic
  subpath. We call the jump $P'$
  \emph{free with respect to $\SSS'$} if $P'$ is not used by any path
  in $\SSS'$.

  We are now ready to complete the proof of the lemma. Note first
  that, as $k>3d$, 
  for any three vertices $b_1, b_2, b_3\in A$ there is a path $P\in
  \SSS$ that contains $b_1, b_2, b_3$. Hence, we can choose an $h \in \{
  1, \dots, k\}$ such that $a_1, a_{\ell} \in V(P_h)$ and there is a
  vertex $a_r$ with $1<r<\ell$ such that $a_r\in V(P_h)$. 
  Let $I' := I\setminus \{ (s_h, t_h)\}$. 
  If $A\subseteq V(P_h)$, then $\SSS \setminus \{ P_h\}$ is a $(k-d-1)$-routing
  of $I'$ and we are done. Otherwise, for every
  vertex $a_r\in A$ which has congestion $k-d$ with respect to $\SSS
  \setminus \{ P_h\}$ there are $i, j$ with $i < r < j$ and an $(i,j)$-jump $P$
  of colour~$h$. This follows as $a_1, a_{\ell}\in V(P_h)$. Note also that  $a_1$
  and $a_{\ell}$ have congestion $k-d-1$ in $\SSS\setminus \{ P_h\}$. 
  Note that this jump $P$ is free with respect to $\SSS\setminus \{
  P_h\}$. 

  Thus, it is easily seen that $\SSS\setminus\{P_h\}$ satisfies the
  following two properties:
  \begin{enumerate}
  \item For every vertex $a_r$ of congestion $k-d$ with
    respect to $\SSS\setminus \{P_h\}$ there are indices $i<r<j$ such
    that there is a free $(i,j)$-jump $P$ with respect to 
    $\SSS\setminus \{P_h\}$. 
  \item For any three vertices $b_1, b_2, b_3$ of congestion $k-d$
    with respect to $\SSS\setminus \{P_h\}$ there is a path $Q\in \SSS\setminus \{P_h\}$ with $\{b_1,
    b_2, b_3\} \subseteq V(Q)$.
  \end{enumerate}
%  Note that
%  this is well defined as if $\LLL$ is a conservative $I'$-routing
% then every vertex of congestion $k-c$ with respect to $\LLL$ must be
%  in $A$ and hence the concept of jumps is well defined for
% $\LLL$. Furthermore, $\SSS\setminus \{ P_c\}$ satisfies all
%  conditions and hence there is such a routing $\SSS'$ as
%  required. \todo{Do we need to explain this further?}

  Now let $\SSS'$ be a routing of $I'$ which satisfies Condition
  $1$ and $2$ (with respect to $\SSS'$ instead of
  $\SSS\setminus\{P_h\}$) and, subject to this, the number of vertices of
  congestion $k-d$ with respect to $\SSS'$ is minimal.

  We claim that $\SSS'$ is a $(k-d-1)$-routing of $I'$. 
  Let $\SSS' := \{ Q_1, \dots, Q_{k-1}\}$.  Towards a
  contradiction, suppose there is a vertex $a_r$ of congestion $k-d$
  with respect to $\SSS'$. As $\SSS'$ is conservative, we have $a_r\in A$. Hence, by
  assumption, there are $i<r<j$ and a free $(i,j)$-jump $P$ with
  respect to $\SSS'$. 

  Let $Q_h$ be a path in $\SSS'$ that contains $a_i, a_r$ and $a_j$, which
  exists by Condition $2$.
  Let $Q_h :=  Q_h^1\cup Q_h^2\cup Q_h^3$ where 
  \begin{itemize}
  \item $Q_h^1$ is the initial subpath of $Q_h$ from its first vertex
    to $a_i$, 
  \item $Q_h^2$ is the subpath starting at $a_i$ and ending in
    $a_j$ and 
  \item $Q_h^3$ is the subpath starting in $a_j$ and ending at
    the end of $Q_h$.
  \end{itemize}

  We define  $Q'_h := Q_h^1 \cup P \cup Q_h^2$, i.e.~$Q'_h$ is the path obtained
  from $Q_h$ by replacing the inner subpath $Q_h^2$ by the
  $(i,j)$-jump $P$.
  Let $\LLL := (\SSS'\setminus \{ Q_h \}) \cup \{ Q'_h \}$. Then
  $\LLL$ is a routing of $I'$. It is also conservative as we have only
  rerouted a single path along a free jump.

  We need to show that for all $b_1, b_2, b_3$ of congestion $k-d$
  with respect to $\LLL$ there is a path $Q\in \LLL$ containing $b_1,
  b_2, b_3$. By assumption, such
  a path $Q'$ exists in $\SSS'$. If $Q' \not= Q_h$, then we are
  done. So suppose $Q_h = Q'$. But then this implies that $b_s\not\in \{ a_{i+1}, \dots, a_{j-1}\}$ for
  all $1\leq s \leq 3$ as otherwise the congestion of $b_s$ would have
  dropped to $k-d-1$ in $\LLL$. But then $b_1, b_2, b_3\in V(Q'_h)$.

  It remains to show that for every vertex $a_s$ of congestion $k-d$
  with respect to $\LLL$ there is a free $(i, j)$-jump for some
  $i<s<j$. 
  As before, by assumption, there are $s_1 < s < s_2$ and a free $(s_1,
  s_2)$-jump with respect to $\SSS'$. If this jump is not $P$, then it
  still exists with respect to $\LLL$ and we are done. So suppose this
  jump is $P$, which implies that $i<s<j$. Furthermore, $a_s\not\in
  Q_h$ as otherwise the congestion of $a_s$ in $\LLL$ would be 
  $k-d-1$. But then, there must be indices $i_1, i_2$ with $i\leq i_1
  < s < i_2\leq j$ such that $a_{i_1}, a_{i_2} \in V(Q_h)$ and
  $a_{s'}\not\in V(Q_h)$ for all $i_1 < s' < i_2$. Hence, the atomic
  subpath $Q''$ of $Q_h$ from $a_{i_1}$ to $a_{i_2}$ is an $(i_1, i_2)$-jump
  as required. As $Q''\subseteq Q_h^2$, this jump is now free. 

  Finally, the vertex $a_r$ now has congestion $k-d-1$ with respect to
  $\LLL$ as $a_r$ is not contained in $Q'_h$. Hence, $\LLL$ has fewer
  vertices of congestion $k-d$ than $\SSS'$,  contradicting the
  choice of $\SSS'$. Thus, $\SSS'$ must have been a $(k-d-1)$-routing
  of $I'$ as required.   This completes the proof of the lemma.
\end{proof}

The previous lemma has the following consequene which essentially
implies Theorem~\ref{thm:main-algo}. 

\begin{corollary}\label{cor:reducing-number-of-paths}
  Let $G$ be an acyclic digraph, let $d\geq 0$ and let $I := \{ (s_1,
  t_1), \dots, (s_k, t_k)\}$ be a set of pairs of vertices such that
  for all $1\leq i \leq k$ there is a path in $G$ linking $s_i$ to $t_i$. $G$
  contains a $(k-d)$-routing of $I$ if, and only if, there is a subset
  $I'\subseteq I$ of order $|I'| \leq 3d$ such that $G$ contains a
  $(k{-}dc)$-routing of $I'$. If $k\geq 3d$, then $I'$ can be chosen of
  size exactly $3d$.
\end{corollary}

We are now ready to prove Theorem~\ref{thm:main-algo}.

\begin{proof}[Proof of Theorem~\ref{thm:main-algo}]
  Let $G, k, d$ and $I := \{(s_1, t_1), \dots, (s_k, t_k)\}$ be
  given. Let $n := |G|$. If for some $1\leq i \leq k$ there is no path
  in $G$ from $s_i$ to $t_i$, then the answer is no and we are
  done. If $k\leq 3d$, then we can apply
  Corollary~\ref{cor:constant-number-of-paths} to compute the answer in
  time $n^{O(d)}$ as required. 

  Otherwise, by Corollary~\ref{cor:reducing-number-of-paths},
  there is a $(k-d)$-routing for $I$ in $G$ if, and only if, there is
  a subset $I'\subsetneq I$ of order $3d$ such that $I'$ has a
  $(3d-d)$-routing. There are ${k \choose 3d} \leq k^{3d} \leq n^{3d}$
  subsets $I'$ of order $3d$. By
  Corollary~\ref{cor:constant-number-of-paths}, we can decide for any
  such $I'$ of order $3d$ in time $n^{O(d)}$ whether a
  $(3d-d)$-routing of $I'$ exists. Hence, by iterating through all
  possible subsets $I'$, we can decide in time $n^{O(d)}$ whether
  there is a $(k-d)$-routing of $I$ in $G$.
\end{proof}

\section{Lower Bounds}\label{sec:lowerbounds}
\newcommand{\partsub}{\textsc{Partitioned Subgraph Isomorphism}\xspace}

In this section, we prove Theorem~\ref{thm:hardness} by a reduction from \partsub.
The input of the \partsub problem consists of a graph $H$ with vertex set
$\{u_1,\dots,u_k\}$ and a graph $G$ whose vertex set is partitioned
into $k$ classes $V_1$, $\dots$, $V_k$. The task is to find a mapping
$\mu:V(H)\to V(G)$ such that $\mu(u_i)\in V_i$ for every $1\le i \le
k$ and $\mu$ is a subgraph embedding, that is, if $u_i$ and $u_j$ are
adjacent in $H$, then $\mu(u_i)$ and $\mu(u_j)$ are adjacent in
$G$. 
\begin{theorem}[\cite{marx-toc-treewidth}]\label{thm:subgraph}
  Assuming ETH, \partsub cannot be solved in time $f(k)n^{o(k/\log
    k)}$ (where $k=|V(H)|$) for any computable function $f$, even when $H$ is assumed to be 3-regular and bipartite.
\end{theorem}

To prove Theorem~\ref{thm:hardness}, we need a reduction from \partsub (for 3-regular bipartite graphs) to \kcrouting{k}{c}, where the number $k$ of demands is linear in the number of vertics of $H$. 
\newcommand{\Pu}{\overline{Q}}
\newcommand{\Pd}{\underline{Q}}
\newcommand{\vu}{\overline{q}}
\newcommand{\vd}{\underline{q}}

% \begin{theorem}\label{thm:hardness}
% For any fixed integer $c\ge 1$, $k$-Congestion Routing is \textup{W[1]}-hard parameterized by $k$ and, assuming ETH, cannot be solved in time $f(k)n^{o(k/\log k)}$ for any computable function $f$.
% \end{theorem}
\begin{proof}[Proof (of Theorem~\ref{thm:hardness})]
  We prove the theorem by a reduction from \partsub. Let $H$ and $G$
  be two graphs, let $V(H)=\{u_1,\dots, u_{k}\}$, and let
  $(V_1,\dots, V_{k})$ be a partition of $V(G)$. By copying vertices
  if necessary, we may assume that every $V_i$ has the same size $n$;
  let us denote by $\{v_{i,1},\dots,v_{i,n}\}$ the vertices in
  $V_i$. By Theorem~\ref{thm:subgraph}, we may assume that $H$ is
  3-regular and bipartite. This means that $H$ has exactly $h=3k/2$
  edges and both partite classes contain $k/2$ vertices. Without loss
  of generality, we can assume that $U_1=\{u_1,\dots,u_{k/2}\}$ and
  $U_2=\{u_{k/2+1},\dots,u_k\}$ are the two partite classes. Let us
  fix an arbitrary ordering $e_1$, $\dots$, $e_{h}$ of the edges of
  $H$.

\textbf{Construction.}  We construct an instance of $(k,c)$-Congestion Routing the following
  way.  We construct a directed graph $D$ that contains, for every
  $1\le i \le k$, two directed paths $\Pu_i$ and $\Pd_i$ (see Figure~\ref{fig:hard}).  Path
  $\Pu_i$ has $n(h+1)+1$ vertices: it contains the vertices $\vu_{i,0}$,
  $\dots$, $\vu_{i,n}$ in this order and additionally, for every
  $1\le j \le n$, the vertices $\vu_{i,j,1}$, $\dots$, $\vu_{i,j,h}$
  are inserted between $\vu_{i,j-1}$ and $\vu_{i,j}$. The path $\Pd_i$
  is defined the same way, with vertices $\vd$ instead of $\vu$. For
  every $1\le \ell \le h$, we introduce two vertices $s_\ell$ and
  $t_\ell$. Then we complete the construction of the graph $D$ by
  introducing further edges as follows.
\begin{figure}[t]
\begin{center}
{\small \svg{\linewidth}{hard}}
\end{center}
\caption{Part of the directed graph $D$ constructed in the proof of Theorem~\ref{thm:hardness} with $k=4$, $h=6$, and $n=5$. For clarity, we consider only one edge $e_4$ of $H$, which connects $u_1$ and $u_3$, and assume that the only edge between $V_1$ and $V_3$ is between $v_{1,3}$ and $v_{3,5}$. The highlighted red paths show the paths $P^v_1$, $P^v_3$, and $P^e_4$ of the solution.}\label{fig:hard}
\end{figure}

\begin{itemize}
\item  For every $1\le i \le k$ and $1\le j \le n$, we introduce the edge $(\vu_{i,j-1},\vd_{i,j})$ (the curved bypass edges in Figure~\ref{fig:hard}).
\item For every $1\le i \le k$, $1\le j \le n$, and $1\le s \le h$, we introduce the edge $(\vu_{i,j,s},\vd_{i,j,s})$ (the vertical edges in Figure~\ref{fig:hard}).
\item For every $1\le \ell \le h$, we do the following. Suppose that edge
  $e_\ell$ of $H$ connects $u_{i_a}$ and $u_{i_b}$ for some $1\le i_a \le k/2$ and
  $k/2+1\le i_b \le k$. Then for every pair of vertices $v_{i_a,j_a}\in V_{i_a}$ and $v_{i_b,j_b}\in V_{i_b}$ that are adjacent in $G$, we introduce the following three edges into $D$: $(s_\ell,\vu_{i_a,j_a,\ell})$, $(\vd_{i_a,j_a,\ell},\vu_{i_b,j_b,\ell})$, and $(\vd_{i_b,j_b,\ell},t_\ell)$.
\end{itemize}
To complete the construction of the $(k,c)$-Congestion Routing instance, we define the following set of $k+2k(c-1)+h$ demands:
\begin{itemize}
\item For every $1\le i \le k$, we introduce the demand $(\vu_{i,0},\vd_{i,n})$ (vertex demands).
\item For every $1\le i \le k$, we introduce $c-1$ copies of the demand $(\vu_{i,0},\vu_{i,n})$ (blocking demands).
\item For every $1\le i \le k$, we introduce $c-1$ copies of the demand $(\vd_{i,0},\vd_{i,n})$ (blocking demands).
\item For every $1 \le \ell \le h$, we introduce the demand $(s_\ell,t_\ell)$ (edge demands). 
\end{itemize}
Note that, for every fixed $c\ge 1$, the number of demands is $O(k)$. In the rest of the proof, we show that a routing with congestion $c$ exists if and only if the \partsub instance has a solution. Then the \textup{W[1]}-hardness and lower bound stated in Theorem~\ref{thm:subgraph} implies the same hardness results for the routing problem.

\textbf{Subgraph embedding $\Rightarrow$ routing.}
Suppose first that vertices $v_{1,z_1}\in V_1$, $\dots$, $v_{k,z_k}\in V_k$ form a solution to the \partsub instance. We construct a routing that contains the following paths, satisfying the demands defined above:

\begin{itemize}
\item For every $1\le i \le k$, the vertex demand $(\vu_{i,0},\vd_{i,n})$ is satisfied by 
a path $P^v_i$ that goes from $\vu_{i,0}$ to $\vu_{i,z_i-1}$ on $\Pu_i$, uses the edge $(\vu_{i,z_i-1},\vd_{i,z_i})$, and the goes from $\vd_{i,z_i}$ to $\vd_{i,n}$ on $\Pd_i$.
\item For every $1\le i \le k$, each of the $c-1$ copies of the blocking demand $(\vu_{i,0},\vu_{i,n})$ is satisfied by a path going on $\Pu_i$.
\item For every $1\le i \le k$, each of the $c-1$ copies of the blocking demand $(\vd_{i,0},\vd_{i,n})$ is satisfied by a path going on $\Pd_i$.
\item For every $1 \le \ell \le h$, the edge demand $(s_\ell,t_\ell)$ is satisfied by a 5-edge path $P^e_\ell=(s_\ell,\vu_{i_a,z_{i_a},\ell},$ $\vd_{i_a,z_{i_b},\ell},\vu_{i_b,z_{i_b},\ell},\vd_{i_b,z_{i_b},\ell},t_\ell)$.
\end{itemize}
It is easy to verify that these are indeed paths: all the required
edges exist. We claim that each vertex of $D$ is used by at most $c$
of these paths. It is easy to see that two paths $P^v_{i'}$ and $P^v_{i''}$ with $i\neq i''$ satisfying vertex
demands do intersect, and this is also true for any two paths $P^e_{\ell'}$ and $P^e_{\ell''}$ with $\ell'\neq \ell''$
satisfying edge demands (note that each vertex of the path $P^e_\ell$  has $\ell$ in its index). The crucial
observation is that the path $P^v_i$ does intersect the path $P^e_\ell$ for any $\ell$. The only way this could possibly happen is
if edge $e_\ell$ of $H$ connects $u_{i_a}$ with $u_{i_b}$, and $i$ is
equal to $i_a$ or $i_b$. But the path $P^e_\ell$
uses only vertex $\vu_{i_a,z_{i_a},\ell}$ from $\Pu_{i_a}$ and vertex
$\vd_{i_a,z_{i_b},\ell}$ from $\Pd_{i_b}$, while the path $P^v_i$  does not use these vertices, as it jumps from
$\vu_{i,z_i-1}$ to $\vd_{i,z_i}$. Thus each vertex is used by at most $c-1$ paths satisfying a blocking demand and at most one additional path satisfying a vertex or edge demand. We can conclude that each vertex is used by at most $c$ of the paths, what we had to show.

\textbf{Routing $\Rightarrow$ subgraph embedding.} Next we show that given a routing with congestion $c$, it is possible to construct the required subgraph embedding from $H$ to $G$. It is clear that the path satisfying the blocking demand $(\vu_{i,0},\vu_{i,n})$ is exactly $\Pu_i$: after leaving $\Pu_i$, there is no way to return back to it. Similarly, the solution must use path $\Pd_i$ to satisfying the blocking demand $(\vd_{i,0},\vd_{i,n})$. It is also clear that the path $P^v_i$ satisfying the vertex demand $(\vu_{i,0},\vd_{i,n})$ has to be contained in the union of $\Pu_i$ and $\Pd_i$. Let $1\le z_i \le n$ be the smallest value such that $\vd_{i,z_i}$ is on path $P^v_i$ (note that this value is positive, as vertex $\vd_{i,0}$ cannot be reached from $\vu_{i,0}$). Observe that path $P^v_i$ uses every vertex of $\Pd_i$ from $\vd_{i,z_i}$ to $\vd_{i,n}$ (as it cannot leave $\Pd_i$). Moreover, since $P^v_i$ does not use the part of $\Pd_i$ from $\vd_{i,0}$ to $\vd_{i,z_i-1}$ by definition, it has to use the part of $\Pu_i$ from $\vu_{i,0}$ to $\vu_{i,z_i-1}$.

  We claim that mapping vertex $u_i$ of $H$ to vertex $v_{i,z_i}$ of
  $G$ is a correct subgraph embedding of $H$ into $G$. To show this,
  suppose that edge $e_i$ of $H$ connects $u_{i_a}$ and $u_{i_b}$
  with $1\le i_a \le k/2$ and $k/2+1\le i_b\le k$; we need to show
  that $v_{i_a,z_{i_a}}\in V_{i_a}$ and $v_{i_b,z_{i_b}}\in V_{i_b}$ are adjacent. Consider
  the path $P^e_\ell$ satisfying edge demand $(s_\ell,t_\ell)$.  By
  construction, the vertex of $P^e_\ell$ after $s_\ell$ has to be on
  the path $\Pu_{i_a}$ and the vertex of $P^e_\ell$ before $t_\ell$
  has to be on $\Pd_{i_b}$. The only way to go from $\Pu_{i_a}$ to
  $\Pd_{i_b}$ is to use an edge of the form
  $(\vd_{i_a,j_a,\ell},\vu_{i_b,j_b,\ell})$: the only way we can leave
  the union of $\Pu_{i_a}$ and $\Pd_{i_a}$ is to enter some $\Pu_i$
  with $k/2+1 \le i \le k$, and there is no edge connecting
  $\Pu_{i_b}$ or $\Pd_{i_b}$ with any $\Pu_i$ with $k/2+1 \le i \le k$ and $i\neq i_b$
  (this is the part of the proof where we use that $H$ is bipartite). We claim that $j_a=z_{i_a}$. 
 If $j>z_{i_a}$, then $\vd_{i_a,j_a,\ell}$ is also used by
  the $c-1$ paths satisfying the blocking demand
  $(\vd_{i_a,0},\vd_{i_a,n})$ and (as we have seen) the path
  $P^v_{i_a}$, contradicting the assumption that the routing has
  congestion $c$. If $j<z_{i_a}$, then there is no way for the path
  $P^e_\ell$ to reach $\vd_{i_a,j_a,\ell}$ from $s_\ell$: each vertex of the path $\Pu_{i_a}$ from $\vu_{i_a,0}$ to $\vu_{i_a,j_a}$
  is used by $c-1$ paths satisfying the blocking demand
  $(\vd_{i_a,0},\vd_{i_a,n})$ and (as shown above) by the path
  $P^v_{i_a}$. This shows $j_z=z_{i_a}$ and a similar argument shows $j_b=z_{i_b}$. Now the existence of the edge
 $(\vd_{i_a,z_a,\ell},\vu_{i_b,z_b,\ell})$ means, by construction, that $G$ contains an edge between $v_{i_a,z_a}\in V_{i_a}$ and $v_{i_b,z_b}\in V_{i_b}$, what we had to show.
\end{proof}
\section{Conclusion}
In this paper we have studied the $\kcrouting{k}{c}$ problem on
acyclic digraphs. It is easy to see that the $n^{O(k)}$ algorithm in \cite{FortuneHW80} for solving the
disjoint paths problem on acyclic digraphs can
be extended to an $n^{O(k)}$ algorithm for $\kcrouting{k}{c}$. As we
proved in Theorem~\ref{thm:hardness}, the $n^{O(k)}$ time
algorithm is essentially best possible with respect to the exponent of
$n$, under the Exponential-Time Hypothesis (ETH). We therefore studied
the extreme cases of relatively high congestion $k-d$ for some fixed
value of $d$. In Theorem~\ref{thm:main-algo} we showed that in this
case we can obtain an $n^{O(d)}$ algorithm on acyclic digraphs,
i.e.~the algorithm only depends on the offset $d$ in
$\kcrouting{k}{k-d}$ but not on the number $k$ of demand pairs. The
proof relied on a reduction argument that shows that as long as $k$ is big
enough compared to $d$, then a demand pair can be eliminated without
changing the answer.

It will be interesting to see whether our result can be extended to
larger classes of digraphs. In particular classes of digraphs of
bounded directed tree width would be a natural target. On such
classes, the $k$-disjoint paths problem can be solved in time
$n^{O(k + w)}$, where $w$ is the directed tree width of the input
digraph (see \cite{JohnsonRST01}). It is conceivable that our results
extend to bounded directed tree width classes and we leave this for
future research.

\bibliographystyle{plainurl}
\bibliography{references}

\end{document}